\documentclass[conference]{IEEEtran}
\IEEEoverridecommandlockouts
\usepackage{amsmath,amssymb,amsfonts,amsthm}
\usepackage{algorithmic}
\usepackage{cite}
\usepackage[T1]{fontenc}
\usepackage{float}
\usepackage{graphicx}
\usepackage{svg}
\usepackage{textcomp}
\usepackage{xcolor}
\usepackage{enumitem}
\usepackage{xspace}
\usepackage{mathtools}

\newtheorem{thm}{Theorem}

\newtheorem{remark}{Remark}

\renewcommand{\bar}{\overline}
\newcommand{\bb}[1]{\mathbb{#1}}
\newcommand{\ee}{\varepsilon}
\newcommand{\narrive}{N_{\text{arrivals}}}
\newcommand{\ndepart}{N_{\text{departures}}}

\newcommand{\set}[1]{\left\{ #1 \right\}}
\newcommand{\etal}{\textit{et al}.\xspace}
\newcommand{\ie}{\textit{i}.\textit{e}.,\xspace}
\newcommand{\eg}{\textit{e}.\textit{g}.,\xspace}

\DeclareMathOperator{\jsq}{JSQ}
\DeclareMathOperator{\slq}{SLQ}

\def\BibTeX{{\rm B\kern-.05em{\sc i\kern-.025em b}\kern-.08em
    T\kern-.1667em\lower.7ex\hbox{E}\kern-.125emX}}

\begin{document}

\title{Service-the-Longest-Queue Among $d$ Choices Policy for Quantum Entanglement Switching
\thanks{This work is supported by QuTech NWO funding 2020–2024 Part I ‘Fundamental Research’, Project Number
601.QT.001-1, financed by the Dutch Research Council (NWO). We further acknowledge support from NWO
QSC grant BGR2 17.269. 
 TV was supported by the EPSRC funded INFORMED-AI
project  EP/Y028732/1.}
}

\author{\IEEEauthorblockN{Guo Xian Yau\IEEEauthorrefmark{1}, Thirupathaiah Vasantam\IEEEauthorrefmark{2}, and Gayane Vardoyan\IEEEauthorrefmark{3}}
\IEEEauthorblockA{\IEEEauthorrefmark{1}\textit{Faculty of Electrical Engineering, Mathematics and Computer Science and QuTech, TU Delft, The Netherlands}\\ 
\textit{\IEEEauthorrefmark{2}Department of Computer Science, Durham University, UK}\\
\textit{\IEEEauthorrefmark{3}Manning College of Information and Computer Sciences, University of Massachusetts, Amherst, USA}}}

\maketitle

\begin{abstract}
An Entanglement Generation Switch (EGS) is a quantum network hub that provides entangled states to a set of connected nodes by enabling them to share a limited number of hub resources. As entanglement requests arrive, they join dedicated queues corresponding to the nodes from which they originate. We propose a load-balancing policy wherein the EGS queries nodes for entanglement requests by randomly sampling $d$ of all available request queues and choosing the longest of these to service. This policy is an instance of the well-known power-of-$d$-choices paradigm previously introduced for classical systems such as data-centers. In contrast to previous models, however, we place queues at nodes instead of directly at the EGS, which offers some practical advantages. Additionally, we incorporate a tunable back-off mechanism into our load-balancing scheme to reduce the classical communication load in the network.
To study the policy, we consider a homogeneous star network topology that has the EGS at its center, and model it as a queueing system with requests that arrive according to a Poisson process and whose service times are exponentially distributed.
We provide an asymptotic analysis of the system by deriving a set of differential equations that describe the dynamics of the mean-field limit and provide expressions for the corresponding unique equilibrium state. 
Consistent with analogous results from randomized load-balancing for classical systems, we observe a significant decrease in the average request processing time when the number of choices $d$ increases from one to two during the sampling process, with diminishing returns for a higher number of choices. We also observe that our mean-field model provides a good approximation to study even moderately-sized systems. 
\end{abstract}


\section{Introduction} \label{sec:intro}
Quantum networks connect quantum-equipped devices to enable distributed applications that are not attainable via classical means alone. Examples include quantum computing in the cloud \cite{Broadbent2009:universal-blind-quantum-computation, Leichtle2021:verifying-bqp-computations-noisy-devices, Jiang2007:distributed-quantum-computing-based-on-small-quantum-registers}; quantum-enhanced sensing \cite{Giovannetti2004:quantum-enhanced-measurements, gottesman2012longer}; and quantum key distribution and conference key agreement \cite{BB84:quantum-cryptography-public-key-distribution-and-coin-tossing, bennett1992quantum, Hahn2020:Anonymous-CKA-in-QN}. While some of these applications are inherently entanglement-based, others can also consume entangled states for tasks such as remote state preparation \cite{Bennett2005:remote-preparation-of-quantum-states} and quantum state and gate teleportation \cite{Benett1993:teleporting-unknown-quantum-state-via-dual-classical-and-EPR-channels, gottesman1999demonstrating}. Entanglement is thus an essential resource whose generation and distribution constitute much of the efforts undertaken within a quantum network.

In fiber optic-based quantum networks, photonic losses increase exponentially with distance \cite{munro2015inside,azuma2023quantum}, rendering communication between quantum nodes infeasible without error correcting codes \cite{Knill-Laflamme-1997:Theory-of-QEC, varnava2006loss} and assistive devices such as quantum repeaters or switches \cite{briegel1998quantum,Dür1999:purification-based-quantum-repeaters, Muralidharan2018:one-way-quantum-repeaters-with-quantum-RS-codes, Vardoyan2021:on-the-sthochastic-analysis-of-an-eds, Gauthier2023:control-architecture-for-EGS-in-QN}. In this work, we study one such device -- an Entanglement Generation Switch (EGS) \cite{Gauthier2023:control-architecture-for-EGS-in-QN, Gauthier2024:on-demand-resource-allocation-algo-for-qn-and-its-performance-analysis} -- that can accommodate a number of connected nodes with bipartite entanglement generation. The EGS assists with this process by facilitating nodes' access to a limited number of shared hardware components (\eg Bell state analyzers (BSAs), as depicted in \cite[Figure 2]{Gauthier2024:on-demand-resource-allocation-algo-for-qn-and-its-performance-analysis}, capable of performing optical Bell state measurements on incoming photons), which we refer to as switch ``resources''. Figure~\ref{fig:EGS-architechture} illustrates an EGS serving $N$ nodes with $m$ resources (the architecture is discussed in detail in Section~\ref{subsec:background-egs}). The EGS resources in principle do not necessitate sophisticated technology such as quantum memories, easing requirements both on cost and fabrication in comparison to memory-equipped counterparts. These properties make the EGS highly relevant to near-term metropolitan-area quantum networks, thus motivating the architecture choice for this study.

Thus far, the EGS has been analyzed in settings where connected nodes are responsible for requesting service.
Here, we propose a novel service mode wherein the onus of request solicitation falls on the switch: the EGS thus queries a fixed number of nodes and assigns a resource module to the node with the largest outstanding number of entanglement requests. 
This operation mode of the EGS is an instance of the so-called power-of-$d$-random-choices paradigm which has seen a wide variety of applications ranging from hashing to virtual circuit routing \cite{Mitzenmacher2001:power-of-two-choices-techniques-and-results}. For brevity, we refer to the act of querying $d$ system components (\eg queues) and selecting one for service as a "$d$-choices policy".  Such policies have proven advantageous as load-balancing techniques in settings like data-centers or computer clusters. Here, an arriving request would ideally be assigned to the least loaded server/compute node, but access to full and up-to-date information about workloads might be unavailable or costly to obtain. Assigning the new task to the least loaded of $d\geq 2$ randomly chosen servers achieves a lower communication cost (relative to querying all servers) while considerably reducing the maximum server load (and therefore the average request response time) even with $d=2$.

A factor contributing to incomplete information at the EGS during decision-making is  node status, namely their readiness to commence entanglement generation. 
Each entanglement generation attempt requires nodes' communication qubits to be available for the coordinated emission of photons whose synchronized arrival is expected at an EGS resource. 
If queues are maintained by the EGS, then a possibility arises that a chosen request cannot be immediately serviced due to the respective nodes' unpreparedness to engage in entanglement generation: \eg a node may already be using its communication qubits to generate entanglement with other nodes, or these qubits may be involved in processing tasks that require their participation (\eg two-qubit gates for hardware platforms such as color centers in diamond \cite{childress2006coherent}). To learn this information, the EGS must communicate with the node-pair that issued the request, and multiple communication messages/rounds may be necessary until the EGS finds a serviceable request. We thus have a practical reason to situate request queues at nodes: when the EGS queries $d$ node-pairs for queue sizes, it receives most up-to-date information, including knowledge of which pairs are primed for entanglement generation. This strategy has the potential to amortize some of the communication delays associated with node status querying.

\begin{figure}[t]
    \centering
    \includegraphics[width=0.3\textwidth]{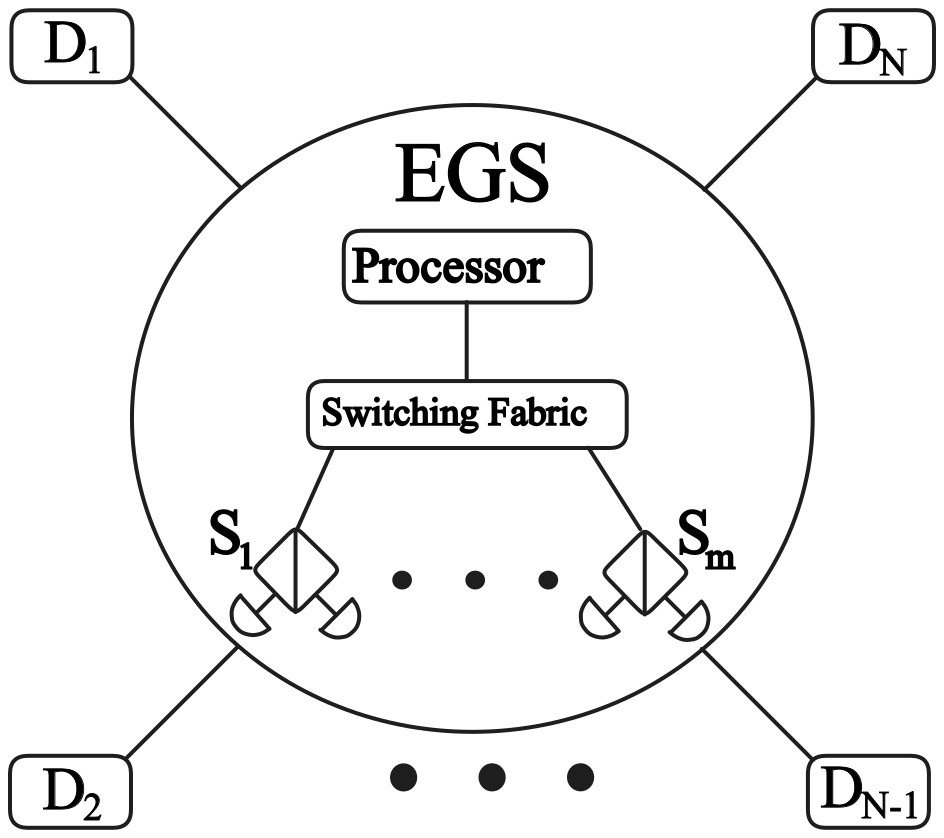}
    \caption{Architecture of an EGS with $m$ resources -- here, BSAs ($S_j$).  $N$ nodes ($D_i$) are connected to the EGS via classical and quantum channels. \vspace{-3mm}}
    \label{fig:EGS-architechture}
\end{figure}

Yet another motivation for placing queues at nodes is that of fairness: the data-center model permits any node to flood the network with its entanglement requests, over time causing the bulk share of the resources to be dedicated to servicing its own demands.
Making the network responsible for offering its services to the nodes removes the latter's ability to directly influence resource allocation within it. We envision that an EGS deploying a $d$-choices policy will provide the aforementioned benefits while at the same time adequately supporting a variety of applications such as distributed quantum computation or entanglement distribution in a high traffic regime -- such an environment ensures that with a high probability, there exists a request within the set of queues sampled by the EGS.

We use mean-field analysis to obtain tight approximations of the average response time -- the time gap between the entry and exit time of an entanglement request. Such analysis has been widely used to develop efficient algorithms for various computer and communication networks \cite{Benam2008ACO}. Due to our design choice of placing queues at nodes, the mean-field limit of the data-center model studied in \cite{Mitzenmacher2001:power-of-two-choices-random-load-balancing}, which has queues at servers, differs from our mean-field limit. Our task in this work is therefore to carry out a performance analysis of the proposed system, which we refer to as the $d$-choices EGS model. Since this study is the first to consider this operation setting of the EGS, we restrict its scope to a single, isolated instance of the device serving nodes in a star topology as depicted in Figure~\ref{fig:EGS-architechture}.
Each node connected to the EGS is assumed to have a number of communication qubits equal to the number of switch resources, enabling it to concurrently participate in entanglement generation with all of them. Because of this assumption, we are able to analyze the model using mean-field techniques. Scenarios where nodes have memory restrictions are reserved for follow-up study. However, we cannot use mean-field techniques for such models; we elaborate more on this in Section~\ref{sec:model-and-mean-field-analysis}.
Our contributions include the following: 
\begin{itemize}[noitemsep,topsep=0pt,leftmargin=*]
    \item We derive a set of differential equations that, in the limit of a large number of request queues, accurately describe the system state evolution of an EGS deploying the $d$-choices policy. This modelling approach enables an asymptotic analysis of the system in the mean-field limit; 

    \item We prove the existence and uniqueness of the equilibrium state that satisfies the mean-field equations;

    \item We provide analytical expressions for request queue size distributions as well as the average response time for an entanglement request. These performance measures are generally challenging to obtain for scenarios where multiple entities (\eg EGS resources) concurrently process requests.
\end{itemize}
Our numerical results support our findings. Namely, they provide supporting evidence that the equilibrium state of the mean-field approximates well the stationary distribution of the system state when the number of request queues is large. Further, we find that the average request response time sees a substantial decrease when $d$ increases from one to two; a further increase in $d$ provides diminishing advantages. We observe approximation errors of less than 5\% for $d = 2$ even for a small system with ten resources and 200 entanglement request queues. We also observe that the approximation error of our model increases as $d$ increases. These results demonstrate the potential of mean-field techniques to study near-term quantum systems where a small number of quantum devices serve a large number of applications. 

The remainder of this manuscript is structured as follows. In Section \ref{sec:background}, we provide relevant background.
In Section~\ref{sec:model-and-mean-field-analysis}, we construct the system model and use it to conduct a mean-field analysis of the system. In Section~\ref{sec:numerical-result}, we present numerical results. We conclude our findings and discuss future directions in Section~\ref{sec:conclusion}.

\section{Background and Related Work}\label{sec:background}
In this section, we describe the EGS architecture and the $d$-choices policy in detail. We then provide an overview of relevant literature. 
\subsection{Entanglement Generation Switch (EGS)} 
\label{subsec:background-egs}

An EGS is a type of quantum entanglement switch (QES) that facilitates entanglement generation between multiple parties. It consists of a central processor tasked with resource scheduling and classical communication with neighboring nodes; resources for entanglement generation -- henceforth we assume for concreteness that these are BSAs although our model is applicable to other types of switch resources; and an optical switching fabric that serves as a bridge between switch interfaces and BSAs (see Figure~\ref{fig:EGS-architechture}). 
\if{false}
\begin{figure}[ht]
    \centering
    \includegraphics[width=0.4\textwidth]{figures/EGS-in-distributed-setting.png}
    \caption{An EGS in the distributed setting. Each quantum device $D_i$ and BSA $S_j$ are connected to the EGS via a classical and a quantum channel.}
    \label{fig:EGS-distributed}
\end{figure}
\fi
When two nodes wish to share a bipartite entangled state (\ie a Bell state/EPR pair \cite{einstein1935can}), they submit a request to the EGS. In this work, we consider the following request handling procedure: when a BSA is available to process a new request, the central processor first informs the nodes that a BSA has been allocated for entanglement generation. Meanwhile, the processor notifies the optical switch that the BSA should only be accessible to said node pair. The optical switch then configures the quantum channel accordingly, and the nodes can proceed with entanglement generation attempts. The result of each attempt is communicated to the nodes classically.

In near-term quantum networks, BSAs in the EGS are likely to be few in number due to device fabrication cost and complexity; hence we assume that they comprise a shared pool of resources as in \cite{Gauthier2023:control-architecture-for-EGS-in-QN} and \cite{Gauthier2024:on-demand-resource-allocation-algo-for-qn-and-its-performance-analysis}. A key challenge for the EGS is thus resource management. While the authors of \cite{Gauthier2023:control-architecture-for-EGS-in-QN} and \cite{Gauthier2024:on-demand-resource-allocation-algo-for-qn-and-its-performance-analysis} managed these resources via request rate control and blocking mechanisms, respectively, our contribution is to study the EGS using the $d$-choices load-balancing policy and to understand queue behavior and the average request processing time. For more details on prior work for the EGS see Section~\ref{subsec:related-work}.

\subsection{The Traditional $d$-Choices Policy} \label{subsec:powd-choices}
The $d$-choices policy is a load-balancing scheme commonly seen in classical data-center models \cite{Mitzenmacher2001:power-of-two-choices-random-load-balancing}. It offers both the benefits of the Join-a-Random-Queue (JRQ) and the Join-the-Shortest-Queue (JSQ) policies; JSQ often has better performance in terms of average request waiting time, whereas JRQ incurs a lower classical communication cost. In a classical data-center with $m$ servers, the $d$-choices policy is referred to as the $\jsq(d)$ policy and is implemented as follows: every server has a queue to store assigned jobs while a job dispatcher assigns an incoming job to the shortest queue size of $d$ randomly chosen servers.
One can easily see that $\jsq(1)$ and $\jsq(m)$ correspond to the JRQ and JSQ, respectively.
One way to analyse $\jsq(d)$ is through a mean-field analysis which we introduce in Section \ref{subsec:mean-field-analysis}. For a more comprehensive summary of techniques developed to study the $d$-choices policy in the data-center setting, we refer the reader to 
\cite{Mitzenmacher2001:power-of-two-choices-techniques-and-results}.

\subsection{Related Work}
\label{subsec:related-work}
Quantum memory-equipped QESes, sometimes referred to as entanglement distribution switches (EDSes) in the literature, have been studied extensively, \eg as in \cite{ Panigrahy2023:capacity-region-QES-with-purification, Vardoyan2021:on-the-sthochastic-analysis-of-an-eds, Vardoyan2021:on-the-exact-analysis-of-an-idealized-quantum-switch, Vardoyan2023:on-the-capacity-region-of-bipartite-and-tripartite-entanglement-switching,Dai2022:the-capacity-region-of-entanglement-switching,nain2020analysis,Avis2023:analysis-of-multipartite-entanglement-distribution-using-a-central-quantum-network-node}.
Architectural differences between EDSes and EGSes give rise to different modes of operation. For instance, the presence of quantum memories in an EDS enables more powerful functionality than that of an EGS, \eg deterministic entanglement swapping \cite{pfaff2013demonstration,Riebe2008:deterministic-entanglement-swapping-with-an-ion-trap-quantum-computer,bhaskar2020experimental}, entanglement distillation \cite{bennett1996purification,deutsch1996quantum,krastanov2019optimized}, and overall more opportunities for dynamic or adaptive decision-making by virtue of being able to store entanglement at the link level (\ie switch-to-node entanglement).
These operational differences result in EDS models that are not directly applicable to our proposed EGS scheme.


Since the limiting resource for an EGS is the BSA, a resource management algorithm is required to facilitate its sharing. 
Gauthier \etal proposed and analyzed a rate-modulation mechanism in \cite{Gauthier2023:control-architecture-for-EGS-in-QN}, as well as a request-blocking mechanism in \cite{Gauthier2024:on-demand-resource-allocation-algo-for-qn-and-its-performance-analysis}. Specifically, the scheme in \cite{Gauthier2023:control-architecture-for-EGS-in-QN} allows node-pairs to submit target entanglement generation rates to the EGS, which are then updated using a rate control protocol to achieve optimal performance according to a metric such as throughput. Our proposed scheme contrasts from this setup in that we assume no control over entanglement demand rates: these are fixed and predetermined in our model.
In \cite{Gauthier2024:on-demand-resource-allocation-algo-for-qn-and-its-performance-analysis}, each node-pair wishing to communicate submits an entanglement request to the EGS. This request is blocked if all BSAs are busy servicing other requests, and the devices will need to resubmit their requests at a later time. This model contrasts from ours in that the former does not allow request queueing.

In this work, we introduce the $d$-choices load-balancing policy as a novel alternative for resource management within an EGS. We use a mean-field limit approach to analyze the system's asymptotic behaviour as the number of queues grows while the server-to-queue ratio stays constant. Previously, a Service-the-Longest-Queue among $d$-choices policy was studied in the context of wireless networks \cite{Murat_Alanyali}. Their model involves one transceiver and $k$ mobile stations, each with a dedicated wireless transmission channel that is available for data transmission only probabilistically due to fluctuations in channel conditions. Each mobile station has a queue to store its packets that are awaiting transmission. The transceiver transmits a packet from one mobile station at a time; the latter must have access to an available wireless channel. Whenever the transceiver becomes idle, it samples $d$ stations at random from the set of all mobile stations with available wireless channels, and chooses the longest of these $d$ queues for processing. The performance of this system was analyzed through mean-field techniques. Our model is different from that of \cite{Murat_Alanyali} as we have multiple servers. In addition, the job service rate and the rate with which an idle server re-samples queues can be different in our model, while they were assumed to be the same in \cite{Murat_Alanyali}. In \cite{Murat_Alanyali}, the authors also studied the policy where the server always processes a job from the longest of all queues with available wireless channels, when the number of queues becomes large. We do not study such a policy as it has a high communication cost.

\section{Model and Mean-field Analysis}\label{sec:model-and-mean-field-analysis}
In this section, we introduce our $d$-choices EGS model and state our assumptions. We then perform a mean-field analysis
by deriving the system's mean-field equations and expressions describing the corresponding equilibrium state. 
\subsection{Model and Assumptions}\label{subsec:model-and-assumptions}

Since we model the EGS as a queueing system, in the following discussion we introduce terminology that will make our subsequent comparison to the data-center model (described in Section~\ref{subsec:powd-choices}) straightforward. Recall that $m$ denotes the number of BSAs within the EGS and that $N$ is the number of nodes connected to the EGS, each with $m$ communication qubits. The BSAs effectively function as \emph{servers}, and we refer to them as such throughout the analysis.  
We define a \emph{flow} $f$ to be a node-pair, $f = (i, j)$, $1 \le i < j \le N$, that desires entanglement. We denote the set of all flows with $\mathcal{F}$, \ie $\mathcal{F} \subseteq \set{(i, j) : 1 \le i < j \le N}$, and $n = |\mathcal{F}|$ the total number of flows. Clearly $1 \le n \le \binom{N}{2}$. 
A \emph{job} is an entanglement request from a flow $f\in \mathcal{F}$. Each job represents the creation of exactly one entanglement between the two nodes of $f$. We assume all jobs of a flow have identical service time distributions -- a reasonable assumption in settings where parameters that drive the entanglement generation process are constant and involuntary parameter drift is insignificant. 


We assume that flow $f$'s jobs arrive according to a Poisson process with parameter $\lambda_f$ and that jobs from different flows arrive independently from each other. Each flow has a queue where jobs await service on a First-Come-First-Served (FCFS) basis. The term \emph{queue size} refers to the number of existing jobs in a queue. We assume successful entanglement generation for flow $f$ has an exponential service time distribution with rate $\mu_f$, which reflects the more realistic view of entanglement generation as a succession of Bernoulli trials with low success probability.
 Since BSAs can operate in parallel and do not affect each other, we assume servers process jobs independently. Finally, we consider a homogeneous system that satisfies $\lambda_f=\lambda$ and $\mu_f=\mu$, $\forall f\in\mathcal{F}$. 


Upon job completion, a (newly idle) server immediately samples $d$ queues at random to acquire a new job to process. Since in our model queues are situated at flows (\ie both nodes of a flow's node-pair keep track of enqueued jobs, but the EGS does not do any active tracking), all $d$ selected queues may be empty. In this event, the server is said to have carried out a \textit{failed sampling}, and remains idle for a period that is exponentially distributed with parameter $\gamma$, before re-sampling $d$ (potentially different) queues. This mechanism exerts a lighter classical communication strain within the system compared to immediate re-sampling. We thus refer to $\gamma$ as the back-off rate. In a near-term EGS system, we expect that entanglement requests will experience a large service time due to photon losses in fiber and failed optical entanglement swaps at BSAs; the latter has a 0.5 success probability without ancilla qubits \cite{grice2011arbitrarily}. It is thus desirable to have $\gamma > \mu$ so that BSAs have shorter idle periods.

Since the EGS serves the longest among $d$ randomly selected queues, we say it follows the $\slq(d)$ policy. Then $\slq(1)$ and $\slq(n)$ are the Service-a-Random-Queue (SRQ) and the Service-the-Longest-Queue (SLQ) policies, respectively. We study the effects of queue placement by comparing the performance of JSQ($d$) and SLQ($d$) in Section~\ref{sec:numerical-result}.
\subsection{Derivation of Mean-field Equations}\label{subsec:mean-field-analysis}
In the analysis that follows, $\bb{N}_{0}$ represents the non-negative integers and $\bb{E}[\cdot]$ is the expectation operator.
With our homogeneity assumption, we denote all arrival rates with $\lambda$ and all service rates as $\mu$. 
We define $r \equiv \frac{m}{n}$ as the server-to-queue ratio.
We next define the following random variables:
\begin{itemize}[noitemsep,topsep=0pt,leftmargin=*]
    \item $\tilde{X}_{i} (t)$, $i \in \bb{N}_{0}$, is the number of flows with at least $i$ jobs at time $t$;
    
    \item $X(t)=(X_{i}(t))_{i \in \bb{N}_{0}}$ where $X_{i}(t)= \frac{1}{n} \tilde{X}_{i} (t)$ is the fraction of flows with at least $i$ jobs at time $t$;


    \item $\tilde{Y}(t)$ is the number of servers servicing a job at time $t$;
    
    \item $Y(t) = \frac{1}{m} \tilde{Y} (t)$ is the fraction of servers servicing a job at time $t$.
\end{itemize}
By convention, we write $\bar{\alpha} = 1 - \alpha$ and $\bar{\alpha}^d = (1 - \alpha)^d$ for $\alpha \in [0, 1]$. For example, $X_{i}^{d}(t) = (X_i(t))^d$, $\bar{Y}(t) = 1 - Y(t)$, and $\bar{X}_{i}^{d}(t) = (1 - X_i(t))^d$. We note that the empirical process $\{(X(t),Y(t))\}_{t\geq 0}$ is a Markov process. 
\begin{remark}
Since each node has $m$ communication qubits, it may participate in entanglement generation at all BSAs simultaneously. Hence, a Markovian representation need not track the identity of each flow currently in service. In the scenario where nodes have fewer than $m$ communication qubits, a Markovian representation must include identities of flows currently being serviced, since in this case $\{(X(t), Y(t))\}_{t\geq 0}$ is not a Markov process. As a result, mean-field analysis that requires the empirical process to be a Markov process is not applicable. The average response time of models with nodes having less than $m$ communication qubits will be lower bounded by our model (nodes have $m$ communication qubits).
\end{remark}

Next we derive mean-field equations without giving a formal proof of the existence of the mean-field limit due to space constraints. This proof follows easily from the theory of the convergence of Markov processes as in \cite{Turner_1998,Mitzenmacher2001:power-of-two-choices-random-load-balancing}.

\begin{thm}
    The \emph{mean-field equations (MFEs)} of the $\slq(d)$ policy applied to the EGS are given by
    \begin{align}
            \frac{d}{dt} x_i (t) &= \lambda (x_{i-1} (t) - x_i (t))\nonumber\\
            &\quad- r (\gamma \bar{y}(t) + \mu y(t)) (\bar{x}_{i+1}^d (t) - \bar{x}_i^d (t)),\label{eq:MFE-queues}\\
        \frac{d}{dt} y(t) &= \gamma \bar{y}(t) (1 - \bar{x}_1^d (t))  - \mu y(t) \bar{x}_1^d (t), \label{eq:MFE-servers}
    \end{align}
    where $i \geq 1$ and $x_0(t)=1$ for all $t \in [0, \infty)$. We refer to \eqref{eq:MFE-queues} as the flow equations and to \eqref{eq:MFE-servers} as the server equation. The process $(x(t),y(t))_{t\geq 0}$ is called the mean-field limit that represents the limit of $\{(X(t),Y(t))\}_{t\geq 0}$ as $n\to\infty$, where $x(t)=(x_i(t),i\geq 0)$.
\end{thm}
An intuitive explanation for the MFEs is as follows: consider first the flow equations \eqref{eq:MFE-queues}. Fix an arbitrary $i \in \bb{N}_{0}$, a sufficiently small $\Delta t > 0$, and let $t \in [0, \infty)$. The following three random variables contribute to change in $\tilde{X}_i(t)$:
\begin{itemize}[noitemsep,topsep=0pt,leftmargin=*]
    \item $U_i$ is the change in $\tilde{X}_i(t)$ due to job arrivals during the period $[t, t + \Delta t]$;
    
    \item $V_i$ is the change in $\tilde{X}_i(t)$ due to a successful re-sampling by an idle server during the period $[t, t + \Delta t]$;
    
    \item $W_i$ is the change in $\tilde{X}_i(t)$ due to a successful sampling by a server upon job completion during $[t, t + \Delta t]$.
\end{itemize}
The expected number of new job arrivals during $[t, t + \Delta t]$ is $n \lambda \Delta t$.
Moreover, the probability that a new job arrives at a queue with $i-1$ jobs is $X_{i-1}(t) - X_i(t)$. Thus,
\begin{align}
    \bb{E}[U_i] 
    & =  n \lambda \Delta t (X_{i-1}(t) - X_i(t)).
\end{align}
Next, consider $V_i$ and $W_i$. When a server selects $d$ queues at random, $(\bar{X}_{i+1}^d(t) - \bar{X}_{i}^d(t))$ is the probability that they all have at most $i$ jobs, with at least one queue having exactly $i$ jobs. We assume that queues are sampled with replacement which is a valid assumption when $n$ is large. Further, the expected number of re-sampling and sampling operations in a $\Delta t$-sized time interval is given by $m \bar{Y}(t) \gamma \Delta t$ and  $m Y(t) \mu \Delta t$, respectively.
We thus obtain
\begin{align}
    \bb{E}[V_i] & =  m \bar{Y}(t) \gamma \Delta t (\bar{X}_{i+1}^d(t) - \bar{X}_{i}^d(t))\label{eq:EV},\\
    \bb{E}[W_i] & =  m Y(t) \mu \Delta t (\bar{X}_{i+1}^d(t) - \bar{X}_{i}^d(t))\label{eq:EW}.
\end{align}

The expected change in $\tilde{X}_i(t)$ during $[t, t + \Delta t]$ is therefore
\begin{align}
      \bb{E} [U_i - V_i &- W_i] = \Delta t \big[ n \lambda (X_{i-1}(t) - X_i(t))\nonumber\\
    & - m (\gamma \bar{Y}(t) + \mu Y(t)) 
     (\bar{X}_{i+1}^d(t) - \bar{X}_{i}^d(t)) \big].
\end{align}
Dividing by $n$ and $\Delta t$, we obtain
\begin{align}
    &\frac{1}{\Delta t} \bb{E}[X_{i} (t+\Delta t)-X_i(t)|\ (X(t),Y(t))] =\nonumber\\
     &\lambda (X_{i-1}(t) - X_i(t))
    - r (\gamma \bar{Y}(t) 
    + \mu Y(t)) (\bar{X}_{i+1}^d(t) - \bar{X}_{i}^d(t)).\label{eq:flow_drift}
\end{align}

To derive \eqref{eq:MFE-servers}, we observe that a change occurs in $Y(t)$ in the interval $[t,t+\Delta t]$  when a busy server completes a job and fails to obtain a new job via $d$-choices sampling, or when an idle server becomes busy via successful $d$-choices sampling. The probability of a failed sampling is given by the probability that all sampled $d$ queues are empty (again assuming sampling with replacement), \ie $\bar{X}_1^d(t)$. Finally, to obtain MFEs we replace $(X(t), Y(t))$ with the mean-field $(x(t), y(t))$ (which is a deterministic process) in the flow and server drift equations, and let $\Delta t \to 0$. 

As a comparison, the MFEs for the $\jsq(d)$ model of \cite{Mitzenmacher2001:power-of-two-choices-random-load-balancing}, which assumes $r=1$,
\begin{equation}
    \frac{d}{dt} x_i(t) = \frac{\lambda}{r} ( x_{i-1}^{d} (t) - x_{i}^{d} (t) ) - \mu (x_i (t) - x_{i+1}(t)). 
\end{equation}
Due to the queue placement in this model, the set of equations above captures the dynamics of the server queues, with the equilibrium point given by $\pi_i = (\frac{\lambda}{r \mu})^{\frac{d^i - 1}{d - 1}}$ for $d \ge 2$. 

\subsection{Equilibrium State of the Mean-Field}
Given the MFEs, we can characterize the system's performance using the equilibrium state, or the fixed point of the mean-field limit. An equilibrium state of the mean-field limit is the solution satisfying $\frac{d y(t)}{dt}=0$ and $\frac{d x_i(t)}{dt}=0$ for all $i\geq 1$. Let us denote by $(\pi, \ee)$ the fixed point satisfying the MFEs, where $\pi = (\pi_i)_{i \ge 0}$ is an infinite sequence with $\pi_i \in [0, 1]$ for all $i \ge 0$ satisfying the flow equations \eqref{eq:MFE-queues} and $\ee \in [0, 1]$ satisfying the server equation \eqref{eq:MFE-servers}.



Next, we characterize the equilibrium state of the mean-field. In the following theorem (Theorem~\ref{thm:fpt}), the equilibrium state of the mean-field exists as long as $\ee=\frac{\lambda}{r\mu}<1$, which leads to $\frac{\lambda}{r(\gamma\bar{\ee}+\mu\ee)}<1$. Note that $\frac{\lambda}{r\mu}<1$ is the necessary condition for the stability of the system. Here, the probability with which a server is busy equals $\frac{\lambda}{r\mu}$ and $r(\gamma\bar{\ee}+\mu\ee)$ is the rate at which a queue is selected for processing when $d=1$. It is of interest to find equilibrium states $(\pi,\ee)$ that satisfy $\sum_{i\geq 1}\pi_i<\infty$ (indicating finite average queue size under $\pi$) as we want to approximate the stationary distribution of a stable system with a finite average queue size.
\begin{thm}\label{thm:fpt}
    Among the class of equilibrium states $(\pi,\ee)$ with $\sum_{i\geq 1}\pi_i<\infty$, if $\frac{\lambda}{r\mu}< 1$ there exists a unique equilibrium state of the mean-field. Furthermore, this equilibrium state satisfies the following recursive equations
    \begin{align}
        \pi_{i+1} & = 1 - \left( 1 - \frac{\lambda}{r (\gamma \bar{\ee} + \mu \ee)} \pi_i\right)^{\frac{1}{d}}, \label{eq:eqm-solution-flows}
    \end{align}
    where $i \ge 1$, $\ee = \frac{\lambda}{r \mu}$,  $\pi_0 = 1$, and $\pi_1 = 1 - \left( 1 - \frac{\lambda}{r (\gamma \bar{\ee} + \mu \ee)} \right)^{\frac{1}{d}}$.
\end{thm}
\begin{proof}
    Since $(\pi, \ee)$ is the equilibrium state of the MFEs, it follows that $\frac{d}{dt} \pi_i = 0$ for all $i \in \bb{N}_{0}$ and $\frac{d}{dt} \ee = 0$. This property, along with \eqref{eq:MFE-queues} and \eqref{eq:MFE-servers}, means that $(\pi,\ee)$ satisfies
    \begin{align}
         \lambda (\pi_{i-1} - \pi_i) - r (\gamma \bar{\ee} + \mu \ee) (\bar{\pi}_{i+1}^d - \bar{\pi}_{i}^d) &= 0, \label{eq:eqm-flows}\\
         \gamma \bar{\ee}(1 - \bar{\pi}_1^d)- \mu \ee\bar{\pi}_1^d &= 0. \label{eq:eqm-servers} 
    \end{align}
    Next, for $i \in \bb{N}_{0}$, let $q_i \coloneqq \pi_i - \pi_{i+1}$; then using \eqref{eq:eqm-flows},
    \begin{align}
        \sum_{i=0}^{j} q_i & = \sum_{i=0}^{j} \frac{r (\gamma \bar{\ee} + \mu \ee)}{\lambda} (\bar{\pi}_{i+2}^d - \bar{\pi}_{i+1}^d) \nonumber \\
        & = \frac{r (\gamma \bar{\ee} + \mu \ee)}{\lambda}( \bar{\pi}_{j+2}^d-\bar{\pi}_1^d) .
        \label{eq:sumqi}
    \end{align}
    On the other hand, $\sum_{i=0}^{j} q_i = \pi_0 - \pi_{j+1} = 1 - \pi_{j+1}$. Using this with \eqref{eq:sumqi} yields
    \begin{equation}
        1 - \pi_{j+1}  = \frac{r (\gamma \bar{\ee} + \mu \ee)}{\lambda}( \bar{\pi}_{j+2}^d-\bar{\pi}_1^d) . \label{eq:partial-sums-Qi}
    \end{equation}
    Since we are interested in equilibrium states that satisfy $\sum_{j\geq 1}\pi_j<\infty$, we use the condition that $\lim_{j \to \infty} \pi_j = 0$ and $\lim_{j \to \infty} \bar{\pi}_{j} = 1$. Applying the limit to both sides of \eqref{eq:partial-sums-Qi} results in
    \begin{align}
        1 
        & = \lim_{j \to \infty} \frac{r (\gamma \bar{\ee} + \mu \ee)}{\lambda}( \bar{\pi}_{j+2}^d-\bar{\pi}_1^d)
         =  \frac{r (\gamma \bar{\ee} + \mu \ee)}{\lambda}( 1-\bar{\pi}_1^d).
    \end{align}
    Rearranging the above equation, we get
    \begin{equation}
        \pi_1 = 1 - \left( 1 - \frac{\lambda}{r (\gamma \bar{\ee} + \mu \ee)} \right)^{\frac{1}{d}}. \label{eq:eqm-solution-pi-1}
    \end{equation}
   By substituting \eqref{eq:eqm-solution-pi-1} into \eqref{eq:eqm-servers}, we obtain $\ee = \frac{\lambda}{r \mu}$. 
Next, we show that \eqref{eq:eqm-solution-flows} is valid. By rearranging \eqref{eq:eqm-flows}, we obtain
    \begin{align}
        \pi_{i+1} & = 1 - \left( \bar{\pi}_i^d + \frac{\lambda}{r (\gamma \bar{\ee} + \mu \ee)} (\pi_{i-1} - \pi_{i})  \right)^{\frac{1}{d}}\label{eq:FP_eq2}. 
    \end{align}
    In \eqref{eq:FP_eq2}, by choosing $i=1$, we obtain
    \begin{align}
        \pi_{2} & = 1 - \left( \bar{\pi}_1^d + \frac{\lambda}{r (\gamma \bar{\ee} + \mu \ee)} (1 - \pi_{1})  \right)^{\frac{1}{d}}.\label{eq:pi_2}     \end{align}
        By substituting \eqref{eq:eqm-solution-pi-1} into $\bar{\pi}_1^d$ of \eqref{eq:pi_2} we get
        \begin{align}
        \pi_{2} & = 1 - \left( 1- \frac{\lambda}{r (\gamma \bar{\ee} + \mu \ee)} \pi_{1}  \right)^{\frac{1}{d}}.     \end{align}
        Similarly, for $i\geq 2$, by expanding the expression for $\bar{\pi}_i^d$ in \eqref{eq:FP_eq2} we obtain \eqref{eq:eqm-solution-flows}.

        Next, we show that the equilibrium state that we found satisfies $\sum_{i\geq 1}\pi_i<\infty$. It suffices to show that $\pi_i\leq \pi^*_i$ by induction on $i \ge 1$, where $(\pi^*,\ee^*)$ is the equilibrium state of the mean-field when $d=1$, and $\pi^*=(\pi^*_i,i\geq 0)$. Furthermore, it can be checked that $\sum_{i\geq 1}\pi^*_i=\frac{\rho}{1-\rho}$ where $\rho=\frac{\lambda}{r(\gamma\bar{\ee}+\mu\ee)}$. Observe that $1-p^{\frac{1}{d}}\leq 1-p$ for any $p \in [0, 1]$ and $d \in \bb{N}$. For the base case, we get
        \begin{align}
            \pi_{1} & = 1 - \left( 1- \frac{\lambda}{r (\gamma \bar{\ee} + \mu \ee)
            }  \right)^{\frac{1}{d}}\\
            & \leq 1 - \left( 1- \frac{\lambda}{r (\gamma \bar{\ee} + \mu \ee)}   \right)\\
            & = \frac{\lambda}{r (\gamma \bar{\ee} + \mu \ee)}=\pi^*_1.
        \end{align}
        For the inductive step, we can rewrite \eqref{eq:eqm-solution-flows} and get
        \begin{align}
            \pi_{i+1} & = 1 - \left( 1- \frac{\lambda}{r (\gamma \bar{\ee} + \mu \ee)} \pi_i \right)^{\frac{1}{d}}\\
            & \leq 1 - \left( 1- \frac{\lambda}{r (\gamma \bar{\ee} + \mu \ee)}\pi_i \right)\\
            & \leq \frac{\lambda}{r (\gamma \bar{\ee} + \mu \ee)}\pi^*_i = \pi^*_{i+1}.
        \end{align}
        Thus, $\pi_{i+1} \le \pi_{i+1}^{*}$ whenever $\pi_{i} \le \pi_{i}^{*}$ as desired.
\end{proof}


It is important to note that we still need to prove that the equilibrium state of the mean-field yields the stationary probability distribution of a queue as $\pi$ and the stationary probability that a server is busy as $\ee$ when $n\to\infty$. A sufficient condition to establish this result is to show the global stability of the mean-field, which we leave for follow-up work. However, as we show in Section \ref{sec:numerical-result}, our numerical results provide supporting evidence that the equilibrium state of the mean-field approximates the stationary distribution of the system when $n$ is large.

\section{Numerical Results}\label{sec:numerical-result}

 In this section, we simulate an EGS deploying the $d$-choices policy and make comparisons to analytical results. The simulation parameters are configured with $\narrive = 10^9$ job arrivals per simulation run, service rate $\mu = 1$,  back-off rate $\gamma = 1$, $N = 21$, $m = 10$, $n = 200 < \binom{N}{2}$ flows/queues for a server-to-queue ratio of $r = \frac{m}{n} = 0.05$, $\frac{\lambda}{r \mu} = 0.90$, and time units are given in seconds, unless otherwise specified. For our simulations, we choose $m \ll n$ as near-term EGSes will have a limited number of resources.

First, we explain how we compute the average response time using simulations. We denote by $t_i$ and $u_i > t_i$ the arrival and departure time of job $i$, respectively. The \emph{response time} of job $i$ is $w(t_i) := u_i - t_i$. The number of jobs departed by time $t > 0$ is denoted by $\ndepart(t)$. The \emph{average response time at time $t > 0$} is given by $\bar{w}(t) = \frac{1}{\ndepart(t)} \sum_{t_i < t} w(t_i)$. The \emph{average response time} over an entire simulation run is $\bar{w} = \bar{w}(t_{\narrive})$, where $t_{\narrive}$ is the arrival time of the last job in the simulation. Let us denote by $\tilde{w}_{\pi}$ the analytically derived average waiting time at equilibrium. By Little's Law \cite[Section 13.7]{vanMieghem2014:performance-analysis-of-complex-networks-and-systems}, we have $\tilde{q}(\pi) = \lambda \tilde{w}_{\pi}$, where $\tilde{q}(\pi)=\sum_{i\geq 1}\pi_i$ is the average queue size at equilibrium, and $\lambda$ is the request arrival rate. Thus, $\tilde{w}_{\pi}=\frac{\tilde{q}(\pi)}{\lambda}$. Given the simulation-based average response time $\bar{w}$ and its analytical approximation $\tilde{w}_{\pi}+1/\mu$, the percentage error between the two values is computed as $100\times \left| 1 - \frac{\tilde{w}_{\pi} + \frac{1}{\mu}}{\bar{w}} \right|\% $. 


\begin{figure}[t]
    \centering
    \includegraphics[width = 0.47\textwidth]{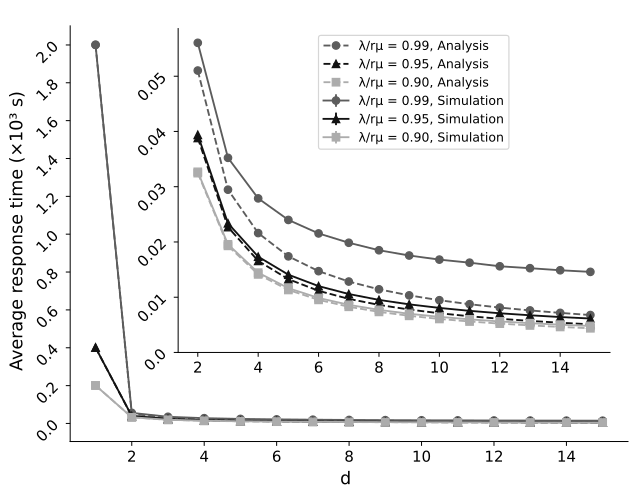}
    \vspace{-3mm}
    \caption{Average response time of an EGS deploying the $d$-choices policy, as a function of the number of choices $d$.}
    \label{fig:avg-resp-time-d-choices}
    \vspace{-4mm}
\end{figure}


\begin{figure}[t]
    \centering
    \includegraphics[width = 0.47\textwidth]{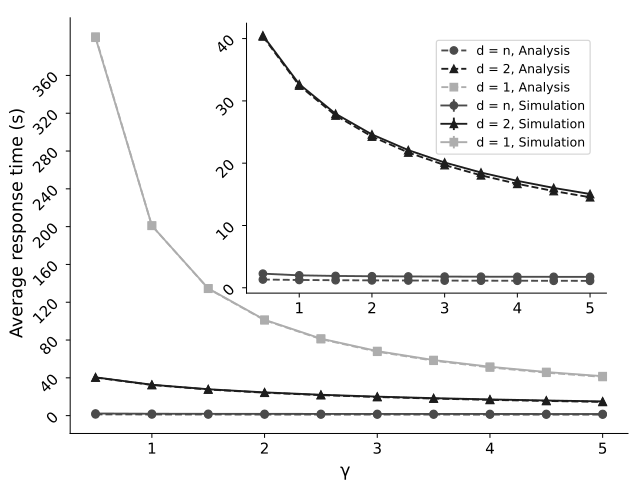}
    \vspace{-3mm}
    \caption{Average response time of an EGS implementing $d$-choices policy as a function of back-off rate $\gamma$.}
    \label{fig:slqd-average-response-time-vs-backoff-rate}
    \vspace{-4mm}
\end{figure}

\begin{figure}[t]
    \centering
    \includegraphics[width = 0.47\textwidth]{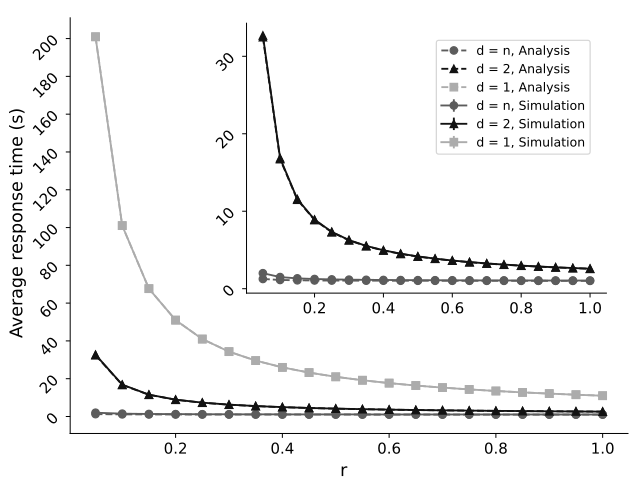}
    \vspace{-3mm}
    \caption{Average response time as a function of server-to-queue ratio $r = \frac{m}{n}$, with varying $m$ and fixed $n = 200$.}
    \label{fig:slqd-average-response-time-vs-server-queue-ratio}
\end{figure}

From Figure~\ref{fig:avg-resp-time-d-choices} we observe that for a fixed $\lambda/r\mu$, increasing the number of choices $d$ reduces the average response time. The effect is most prominent when increasing $d$ from one to two; a further increment in $d$ exhibits a diminishing gain in performance. The decrease in $\bar{w}$ for higher values of $d$ is expected since the EGS is more likely to select a non-empty queue with large queue sizes and perform a successful sampling. In Figure~\ref{fig:slqd-average-response-time-vs-backoff-rate}, we observe that the average response time decreases as the back-off rate $\gamma$ increases. A higher back-off rate allows an idle server to re-sample more frequently after a failed sampling, thus decreasing the server's overall idle time. In Figure~\ref{fig:slqd-average-response-time-vs-server-queue-ratio}, we compare the average response time of $\slq(d)$ as a function of server-to-queue ratio $r = \frac{m}{n}$ with a fixed number of flows (\ie $n = 200$). As servers process the jobs independently, increasing the number of servers $m$ increases the overall service capacity of the EGS. Hence, the overall average response time decreases.

\begin{figure}[t]
    \centering
    \includegraphics[width = 0.47\textwidth]{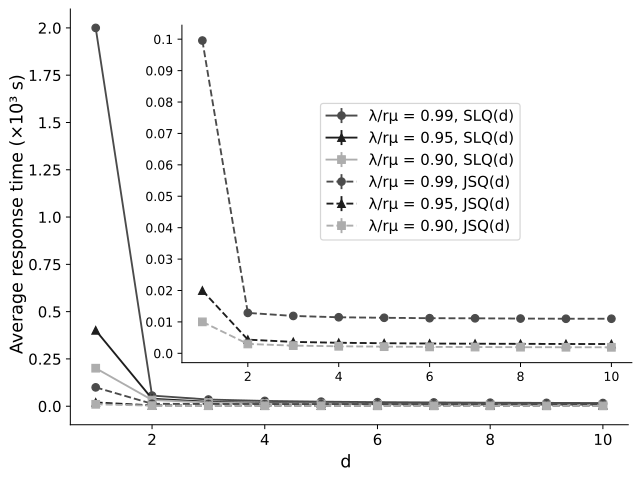}
    \vspace{-3mm}
    \caption{Average response time of $\jsq(d)$ and $\slq(d)$ as a function of the number of choices $d$.}
    \label{fig:slqd-vs-jsqd-average-response-time-powd}
    \vspace{-4mm}
\end{figure}

\begin{figure}[t]
    \centering
    \includegraphics[width = 0.47\textwidth]{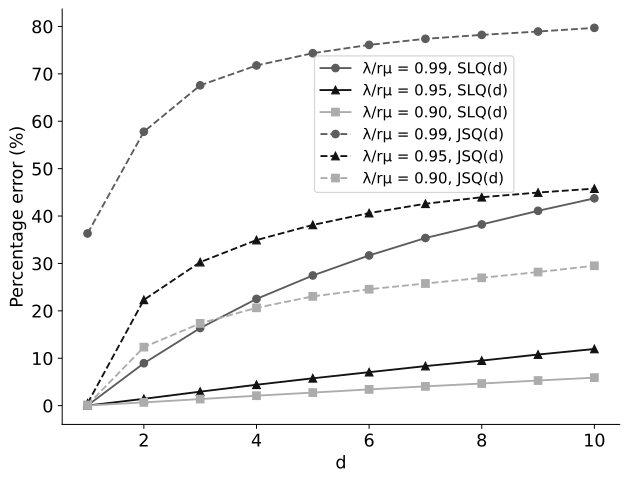}
    \vspace{-3mm}
    \caption{Percentage error of the average response time of $\jsq(d)$ and $\slq(d)$ as a function of the number of choices $d$.}
    \label{fig:slqd-vs-jsqd-percentage-error-powd}
\end{figure}

We now compare the performance of $\slq(d)$ with that of $\jsq(d)$. From Figure~\ref{fig:slqd-vs-jsqd-average-response-time-powd} we observe that $\jsq(d)$ has a lower average response time for a fixed $\frac{\lambda}{r \mu}$ ratio. The reason for this is that in the $\slq(d)$ model, BSAs may experience additional idle time due to failed (re-)sampling. In contrast, in the $\jsq(d)$ model new jobs immediately join the least loaded amongst $d$ randomly chosen queues, and there is never a need to re-sample. Figure~\ref{fig:slqd-vs-jsqd-percentage-error-powd} illustrates that $\slq(d)$ has a lower percentage error for the same $\frac{\lambda}{r \mu}$ ratio. Mean-field analysis necessitates $\slq(d)$ (resp. $\jsq(d)$) to have a large $n$ (resp. $m$); our system parameters do not meet the large $m$ requirement. Figure~\ref{fig:density-function} portrays queue state probability density functions, obtained via simulation as well as using the equilibrium state of the mean-field. While the mean-field approximation becomes less accurate for larger $\frac{\lambda}{r \mu}$ values (\eg at $\frac{\lambda}{r \mu} = 0.99$), overall we observe that our model exhibits close correspondence to the simulation.

\begin{figure}[t]
    \centering
    \includegraphics[width = 0.47\textwidth]{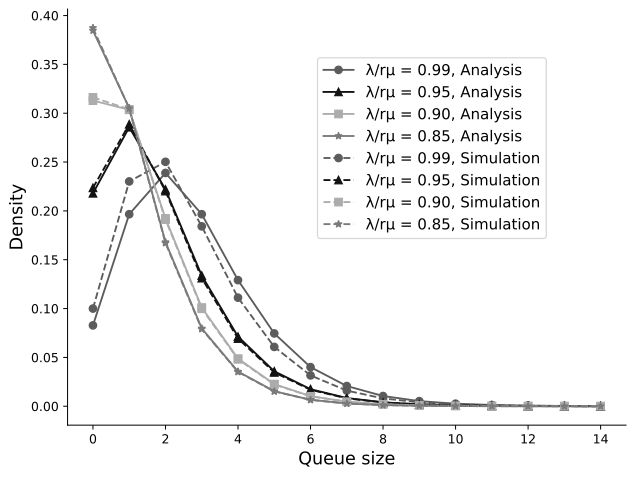}
    \vspace{-3mm}
    \caption{Comparison of analytical and simulated queue state density functions for an EGS deploying the $\slq(2)$ policy.}
    \label{fig:density-function}
    \vspace{-3mm}
\end{figure}

\section{Conclusion}\label{sec:conclusion}

Motivated by the importance of algorithm design and performance analysis of quantum switches with multiple entanglement swapping devices, we proposed the $d$-choices policy for an EGS and used mean-field techniques to study its performance. To this end, we developed a model for a homogeneous star network topology to study the role of this device in near-term quantum networks. For practical reasons, we placed queues at service-requesting nodes, instead of directly at the EGS. In this way, our model contrasts from traditional $d$-choices policies applied within classical data-center models, necessitating new analysis. We then derived a set of differential equations describing the system's evolution in the mean-field limit, and proved the existence and uniqueness of the equilibrium state. Our numerical results show that the mean-field limit is useful to study an EGS even when it has a small number of resources (\eg 10) provided there is a large number of request queues. Particularly, for $\slq(d)$ with $d=2$, mean-field approximations are tight except when $\frac{\lambda}{r\mu}$ is very close to one. Our analysis also applies to classical systems in which servers sample queues of flows to obtain a job for processing.

Our work serves as a preliminary study of an EGS deploying the $d$-choices policy. A valuable follow-up contribution would be to show global stability of the mean-field, and prove that its equilibrium state approximates the system's stationary distribution as the number of queues increases.
The homogeneity assumption can be relaxed by introducing different classes of flows with identical arrival and service rates. Mean-field analysis requires each class to be adequately populated, thus the EGS service model may require alterations to ensure stability.
Our study concentrated on a single, isolated EGS within a star topology. A valuable extension would be to explore multiple EGSes within more complex topologies.

\bibliographystyle{IEEEtran}
\bibliography{main}

\end{document}